\documentclass[10pt]{article}
\usepackage{amsmath}
\usepackage{amssymb, amscd, amsthm}
\usepackage[all]{xy}
\usepackage[dvips]{graphicx}
\usepackage{verbatim}
\usepackage[perpage,symbol*]{footmisc}
\newtheorem{theorem}{Theorem}

\newtheorem{lemma}{\textbf{Lemma}}
\newtheorem{remark}{\textbf{Remark}}

\newtheorem{example}{\textbf{Example}}

\usepackage[justification=centering]{caption}
\usepackage{longtable}

\usepackage{multirow}
\usepackage{float}

\newcommand{\Tr}{{{\rm Tr}}}
\begin{document}

\baselineskip 17pt
\title{\Large\bf Sequences with thirteen-valued cross correlations}

\author{Yuehui Cui\qquad\qquad Jinquan Luo}
\footnotetext{Yuehui Cui and Jinquan Luo are with School of Mathematics
and Statistics \& Hubei Key Laboratory of Mathematical Sciences, Central China Normal University, Wuhan China 430079.\\
 E-mail: hfcyh1@163.com(Y.Cui), luojinquan@ccnu.edu.cn(J.Luo)}

\date{}
\maketitle
\begin{abstract}
In this paper, we completely determine the cross correlation distribution between an $m$-sequence $(s_t)$ of period $p^n-1$ and its $d$-decimated sequence $(s_{dt})$, where $d = \frac{p^n-1}{3} + p^i$,  $p \equiv 1 \pmod{3}$, $\frac{1}{3}p^{-i}(p^n-1) \not\equiv 2 \pmod{3}$, and $0 \leq i < n$.  It is shown that the cross correlation is $13$-valued. To the best of our knowledge, this is the first time that the cross correlation distribution of so many values has been determined. 
\end{abstract}
{\bf Key words}:  Cross correlation,  Exponential sum,  Gaussian period, Cyclic code.

\section{Introduction}

Sequences that exhibit good cross correlation properties are highly valuable in applications, such as radar systems, code division multiple access (CDMA) communications systems, cryptography and error correcting codes. For foundational knowledge and applications of sequences, we refer to \cite{H2, niho, cdma, ROSEPHD}. Let $s=(s_t)_{t=0}^{l-1}$ and $s'=(s'_t)_{t=0}^{l-1}$ be two sequences of the same length $l$. The cross correlation function between $s$ and $s'$ by a shift $\tau$, $0\leq\tau<l$ is defined by
\begin{equation}\label{dfC_d}
C_{\tau}(s,s')=\sum\limits_{t=0}^{l-1}\zeta_p^{s_{t+\tau}-s'_t},
\end{equation}
where $\zeta_p=e^{\frac{2\pi\sqrt{-1}}{p}}$ is a primitive complex $p$-th root of unity. In many cases the sequence $s'$ is a decimated sequence of $s$ with some decimation $d$, that is, $s'_t=s_{(dt)}$ for $0\leq t\leq l-1$ where $(dt)$ takes on the smallest non-negative integer equal to $dt$ module $l$. If \( d \) and \( l \) are coprime, then \( s' \) has the same length as \( s \), and both are called sequences of length \( l \). 

Over the past 50 years, research on cross correlation of sequences has garnered significant interest from different aspects (see \cite{H2, L2, xia1, xia2, xiong}). However, determining the value distribution of the cross correlation function remains theoretically challenging. The reader can refer to two recent survey papers \cite{survey1} and \cite{survey2} for those known decimations.

In \cite{H2}, Helleseth determined the value distribution of the cross correlation function for the Niho exponent \(\frac{p^n-1}{3} + p^i\) under the conditions \(n \equiv 0 \pmod{2}\), \(p \equiv 2 \pmod{3}\), \(\frac{1}{3}p^{-i}(p^n-1) \not\equiv 2 \pmod{3}\), and \(0 \le i < n\). A natural question arises: what happens when \(p \equiv 1 \pmod{3}\) or $n$ is odd? In this paper, we address this case and present the corresponding results. We compute the value distribution of the cross correlation function for the same exponent \(d\) but with \(p \equiv 1 \pmod{3}\) and  \(n\) is an arbitrary positive integer. Our results show that the cross correlation function takes exactly \(13\) distinct values.

Throughout this paper, we always assume that
$
d=\frac{p^{n}-1}{3}+p^i,
$
where $p\equiv 1\,\,({\rm mod}\,\,3)$, $\frac{1}{3}p^{-i}(p^n-1)\not\equiv 2\,\,({\rm mod}\,\,3)$ and  $0\leq i<n$. We can try to study cross correlation of the sequence
\begin{equation*}
  s=(s_t)_{t=0}^{p^{n}-2}=\big(\Tr_1^n(\psi^t)\big)_{t=0}^{p^{n}-2}
\end{equation*} 
and its decimated sequence 
\begin{equation*}
  s'=(s'_{t})_{t=0}^{p^{n}-2}=(s_{dt})_{t=0}^{p^{n}-2}
  =\big(\Tr_1^n(\psi^{dt})\big)_{t=0}^{p^{n}-2}.
\end{equation*}
Then by (\ref{dfC_d}), the cross correlation function with shift $\tau$ is
\begin{equation*}
    C_{\tau}(s,s')=\sum\limits_{t=0}^{p^n-2}\zeta_p^{\Tr ^n_1\left(\psi^{t+\tau}-\psi^{dt}\right)}
    =\sum\limits_{x\in\mathbb{F}_{p^n}^*}\zeta_p^{{\rm{Tr}}_1^n(ux-x^{d})}
\end{equation*}
where $u=\psi^{\tau}$.  The main result of this paper is presented as follows.

\begin{theorem}\label{13distribution} \rm
Let $d=\frac{p^{n}-1}{3}+p^i$,
where $p\equiv 1\,\,({\rm mod}\,\,3)$, $\frac{1}{3}p^{-i}({p^{n}-1})
\not\equiv 2\,\,({\rm mod}\,\,3)$ and $0\leq i<n$. Then $\gcd(d,p^n-1)=1$. The cross correlation distribution $C_{\tau}(s,s')$ when $\tau$ runs from $0$ to $p^n-2$ can be described as follows. 

(i) For case $\frac{1}{3}p^{-i}({p^{n}-1})\equiv 1\,\,({\rm mod}\,\,3)$, Table \ref{13 table1} holds.

(ii) For case $\frac{1}{3}p^{-i}({p^{n}-1})\equiv 0\,\,({\rm mod}\,\,3)$, Table \ref{13 table2} holds.

\begin{table}[H]\footnotesize
\centering
 \caption{ $\frac{1}{3}p^{-i}({p^{n}-1})\equiv 1\,\,({\rm mod}\,\,3)$ }
\label{13 table1}
\begin{tabular}{l l }
\hline
Value & Frequency \\
[0.5ex]
\hline
-1 & $\frac{2}{9}\big(p^n+\frac{{(-1)}^{n-1}}{2^{n-1}}E(u,v,n)-8\big)$ \\
$(-1)^{n+1}2p^{n/2} \cos\left( \frac{n\theta}{3} - \frac{2\pi l}{3} \right)-1$
& $\frac{1}{27}\big(p^n+\frac{{(-1)}^{n-1}}{2^{n-1}}E(u,v,n)-8\big)$ for any $l=0,1,2$  \\ 
$\frac{2\sqrt{3}(-1)^{n+1}p^{n/2}}{3} \cos\left( \frac{n\theta}{3} -\frac{\pi(4l+1)}{6} \right) - 1$
 & $\frac{1}{9}\big(p^n+\frac{{(-1)}^{n}}{2^{n}}\big(E(u,v,n)-O(u,v,n)\big)-2\big)$ for any $l=0,1,2$  \\ 
$\frac{2\sqrt{3}(-1)^{n+1}p^{n/2}}{3} \cos\left( \frac{n\theta}{3} -\frac{\pi(4l-1)}{6} \right) - 1$ 
&   $\frac{1}{9}\big(p^n+\frac{{(-1)}^{n}}{2^{n}}\big(E(u,v,n)+O(u,v,n)\big)-2\big)$ for any $l=0,1,2$  \\
 $(-1)^{n+1}\frac{4p^{n/2}}{3} \cos\left( \frac{n\theta}{3} - \frac{2\pi l}{3} \right)+\frac{p^n-3}{3}$ &  1 for any $l=0,1,2$\\
\hline
\end{tabular}
\end{table}
\begin{table}[H]\footnotesize
\centering
 \caption{ $\frac{1}{3}p^{-i}({p^{n}-1})\equiv 0\,\,({\rm mod}\,\,3)$ }
\label{13 table2}
\begin{tabular}{l l }
\hline
Value & Frequency \\
[0.5ex]
\hline
-1 & $\frac{1}{9}\big(2p^n+2\frac{{(-1)}^{n-1}}{2^{n-1}}E(u,v,n)-25\big)$ \\
$(-1)^{n+1}2p^{n/2} \cos\left( \frac{n\theta}{3} - \frac{2\pi l}{3} \right)-1$
& $\frac{1}{27}\big(p^n+\frac{{(-1)}^{n-1}}{2^{n-1}}E(u,v,n)+1\big)$ for any $l=0,1,2$  \\ 
$\frac{2\sqrt{3}(-1)^{n+1}p^{n/2}}{3} \cos\left( \frac{n\theta}{3} -\frac{\pi(4l+1)}{6} \right) - 1$
 & $\frac{1}{9}\big(p^n+\frac{{(-1)}^{n}}{2^{n}}\big(E(u,v,n)-O(u,v,n)\big)-2\big)$ for any $l=0,1,2$  \\ 
$\frac{2\sqrt{3}(-1)^{n+1}p^{n/2}}{3} \cos\left( \frac{n\theta}{3} -\frac{\pi(4l-1)}{6} \right) - 1$ 
 &   $\frac{1}{9}\big(p^n+\frac{{(-1)}^{n}}{2^{n}}\big(E(u,v,n)+O(u,v,n)\big)-2\big)$ for any $l=0,1,2$  \\
$\frac{2(-1)^{n+1}p^{n/2}}{3} \cos\left( \frac{n\theta}{3} - \frac{\pi(2l+1)}{3} \right) +\frac{p^n-3}{3}$ &  1 for any $l=0,1,2$\\
\hline
\end{tabular}
\end{table}
Here 
\begin{equation*}
\theta = \operatorname{sgn}(v) \cdot \arccos\frac{u}{2\sqrt{p}},
\end{equation*}
\begin{equation*}
  E(u,v,n):=u^n-\sum_{t=1,v_2(t)=1}^n\binom{n}{t}u^{n-t}v^t3^{\frac{t}{2}}
  +\sum_{t=1,v_2(t)\geq2}^n\binom{n}{t}u^{n-t}v^t3^{\frac{t}{2}},
\end{equation*}
\begin{equation*}
  O(u,v,n):=\sum_{t=1,v_2(t+1)=1}^n\binom{n}{t}u^{n-t}v^t3^{\frac{t+1}{2}}
  -\sum_{t=1,v_2(t+1)\geq2}^n\binom{n}{t}u^{n-t}v^t3^{\frac{t+1}{2}},
\end{equation*}
 and the integers $u$ and $v$ are uniquely determined such that 
\begin{equation*}
  u^2+3v^2=4p, u\equiv 1\,\,({\rm mod}\,\,3), v\equiv 0\,\,({\rm mod}\,\,3) \mbox{ and } 3v\equiv u(2\psi^{(p^n-1)/3}+1)\,\,({\rm mod}\,\,p). 
\end{equation*}
Here \(\psi\) is a primitive element of \(\mathbb{F}_{p^n}\) and \(v_2(m)\) is the largest integer \(t\) such that \(2^t \mid m\).
\end{theorem}

\section{Preliminaries}\label{Preliminaries}

In this section, we briefly recall some definitions and results which will be used later in this paper. We begin this section by fixing some notations throughout this paper unless otherwise stated.

\begin{tabular}{ll}
$\mathbb{F}_{p^n}$ & finite field of $p^n$ elements \\
$\mathbb{F}_{p^n}^*$ & multiplicative group of nonzero elements of $\mathbb{F}_{p^n}$ \\
$\psi$ & a primitive element of $\mathbb{F}_{p^n}$ \\
$\Tr^n_1$ & trace function from $\mathbb{F}_{p^n}$ to $\mathbb{F}_p$ \\
$\zeta_p$ &  a primitive complex $p$-th root of unity $e^{2\pi \sqrt{-1}/p}$ \\
$\operatorname{ind}_{\psi}(x)$ & $\operatorname{ind}_{\psi}(x)=t$ for $0 \le t < p^n - 1$, $x\in\mathbb{F}_{p^n}^*$, $x = \psi^t$ \\
$\langle\psi^3\rangle$ & subgroup of $\mathbb{F}_{p^n}^*$ generated by $\psi^3$  \\
$\mathcal{D}_j$  &  $\mathcal{D}_j=\psi^j\langle\psi^3\rangle$\\
$v_p(m)$ &  largest integer $t$ such that $p^t \mid m$, i.e., $v_p(m) = t$\\
$C_{i,j}$ & $\{x\in\mathbb{F}_{p^{n}}\setminus \{0, -1\}: \operatorname{ind}_{\psi}(x+1)\equiv i\pmod{3},\ \operatorname{ind}_{\psi}(x)\equiv j\pmod{3}\}$ \\
$\eta_u$ & Gaussian period $\eta_u=\frac{1}{3}\sum\limits_{y\in \mathbb{F}_{p^n}^*}\zeta_p^{\Tr_1^n(uy^3)}$ for $u\in\mathbb{F}_{p^n}$ \\
$W_d(u,v)$ & $W_d(u,v)=\sum\limits_{x\in\mathbb{F}_{p^n}}\zeta_p^{{\rm{Tr}}_1^n(ux-vx^d)}$ for $(u,v)\in \mathbb{F}_{p^n}\times\mathbb{F}_{p^n}$\\
\end{tabular}

\begin{lemma}\label{p 1 3 squ}\rm
\cite[Theorem 1.2]{FRQ}
Let $p\equiv 1\,\,({\rm mod}\,\,3)$ be a prime. Let $n$ be a positive integer. Let $\psi$ be a primitive element of $\mathbb{F}_{p^n}$. For $a_1$, $a_2$, $c$ $\in\mathbb{F}_{p^n}^*$, denote by $N_1$ the number of solutions $(x_1,x_2)$ in $\mathbb{F}_{p^n}$ of the equation
\begin{equation*}
  a_1x_1^3+a_2x_2^3=c.
\end{equation*}
Then
\begin{equation*}
    N_1=\left\{ \begin{array}{llll}
	p^n+\frac{{(-1)}^{n-1}}{2^{n-1}}E(u,v,n)+\delta(a_1,a_2),&\mbox{if}& ind_{\psi}(a_1a_2c) \equiv 0\,\,({\rm mod}\,\,3),\\
    p^n+\frac{{(-1)}^{n}}{2^{n}}\big(E(u,v,n)-O(u,v,n)\big)+\delta(a_1,a_2),&\mbox{if}&
    ind_{\psi}(a_1a_2c) \equiv 1\,\,({\rm mod}\,\,3),\\ p^n+\frac{{(-1)}^{n}}{2^{n}}\big(E(u,v,n)+O(u,v,n)\big)+\delta(a_1,a_2),&\mbox{if}&
    ind_{\psi}(a_1a_2c) \equiv 2\,\,({\rm mod}\,\,3),
	\end{array}\right.
\end{equation*}
where 
\begin{equation*}
    \delta(a_1,a_2):=\left\{ \begin{array}{llll}
	-2,&\mbox{if}& ind_{\psi}(a_1a_2^2) \equiv 0\,\,({\rm mod}\,\,3),\\
    1,&\mbox{if}&
    ind_{\psi}(a_1a_2^2)\not \equiv 0\,\,({\rm mod}\,\,3),
	\end{array}\right.
\end{equation*}
 and $E(u,v,n)$, $O(u,v,n)$, $u$, $v$ are defined as in Theorem \ref{13distribution}.
\end{lemma}

\begin{lemma}\label{crslemma}\rm
Let \(p\) be an odd prime with \(p \equiv 1 \pmod{3}\), and let \(n\) be a positive integer. Then
\begin{equation*}
    \left \{
\begin{array}{ll}
|C_{0,0}| =\frac{1}{9}\big(p^n+\frac{{(-1)}^{n-1}}{2^{n-1}}E(u,v,n)-8\big), \\
|C_{1,2}| =|C_{2,1}| =\frac{1}{9}\big(p^n+\frac{{(-1)}^{n-1}}{2^{n-1}}E(u,v,n)+1\big), \\
|C_{0,1}|= |C_{1,0}|=|C_{2,2}| =\frac{1}{9}\big(p^n+\frac{{(-1)}^{n}}{2^{n}}\big(E(u,v,n)-O(u,v,n)\big)-2\big),\\
|C_{0,2}| =|C_{2,0}|=|C_{1,1}| =\frac{1}{9}\big(p^n+\frac{{(-1)}^{n}}{2^{n}}\big(E(u,v,n)+O(u,v,n)\big)-2\big).
\end{array}\right.
\end{equation*}
The notations \(E(u,v,n)\), \(O(u,v,n)\), and the integers \(u,v\) are the same as in Theorem \ref{13distribution}.
\end{lemma}
\begin{proof}
In order to compute $|C_{0,0}|$, we need to count the number of solutions to the system
\begin{equation}\label{001b}
    \left \{
\begin{array}{ll}
y_1-y_2=1,\\
\operatorname{ind}_{\psi}(y_1)=\operatorname{ind}_{\psi}(y_2)\equiv 0\,\,({\rm mod}\,\,3).
\end{array}\right.
\end{equation}
The condition $\operatorname{ind}_{\psi}(y_1)=\operatorname{ind}_{\psi}(y_2)\equiv 0\pmod{3}$ implies that we may set
$y_1 = \psi^{3a}$ and $y_2 = \psi^{3b}$, where $a, b \in \{0, 1, \dots, \frac{p^n-4}{3}\}$.
Substituting into $y_1 - y_2 = 1$, we obtain
\begin{equation}\label{3a3b1}
\psi^{3a} - \psi^{3b} = 1.
\end{equation}
Let $x = \psi^{a}$ and $y = \psi^{b}$. Then $x, y \in \mathbb{F}_{p^{2m}}^*$, and the equation becomes
\begin{equation}\label{bu1}
x^{3} - y^{3} = 1.
\end{equation}
Applying Lemma \ref{p 1 3 squ} with $a_1 =1 $, $a_2 = \psi^{\frac{p^n-1}{2}}$ and $c= 1$, the number of solutions to (\ref{bu1}) in $\left(\mathbb{F}_{p^{n}}\right)^2$  is
\[
N_1 = p^n+\frac{{(-1)}^{n-1}}{2^{n-1}}E(u,v,n)-2.
\]
Observe that if $(x_0, y_0)$ is a solution to \eqref{bu1}, then for any $k_1, k_2 \in \{0, 1, 2\}$,
\(
\left(x_0 \psi^{k_1(p^n-1)/3}, y_0 \psi^{k_2(p^n-1)/3}\right)
\)
is also a solution. The equation \eqref{bu1} has  six solutions with \(x = 0\) or \(y = 0\). Therefore, the number of solutions to \eqref{bu1} in \((\mathbb{F}_{p^{n}}^*)^2\) is
\[
N'_1 = p^n + \frac{(-1)^{n-1}}{2^{n-1}}E(u,v,n) - 8.
\]
Since each pair $(a,b)$ satisfying (\ref{3a3b1}) corresponds to $9$ distinct pairs $(x,y)\in\left(\mathbb{F}_{p^{n}}\right)^2$ satisfying (\ref{bu1}), the number of solutions to the original system \eqref{001b} is
\[
\frac{N'_1}{9} = \frac{1}{9}\left(p^n+\frac{{(-1)}^{n-1}}{2^{n-1}}E(u,v,n)-8\right).
\]
Hence, $|C_{0,0}| = \frac{1}{9}\left(p^n+\frac{{(-1)}^{n-1}}{2^{n-1}}E(u,v,n)-8\right)$.

In order to compute $|C_{i,j}|$, we need to count the number of solutions to the system
\begin{equation}\label{001b11}
    \left \{
\begin{array}{ll}
\psi^i y_1-\psi^jy_2=1,\\
\operatorname{ind}_{\psi}(y_1)=\operatorname{ind}_{\psi}(y_2)\equiv 0\,\,({\rm mod}\,\,3).
\end{array}\right.
\end{equation}
For the remaining cases, a similar analysis yields results.
\end{proof}

\section{Proof of Theorem~\ref{13distribution}}

In this section we complete the proof of Theorem~\ref{13distribution}. First, we begin with two lemmas.

\begin{lemma}\label{wf21value} \rm
Let $d=\frac{p^{n}-1}{3}+p^i$,
where $p\equiv 1\,\,({\rm mod}\,\,3)$, $\frac{1}{3}p^{-i}({p^{n}-1})
\not\equiv 2\,\,({\rm mod}\,\,3)$ and $0\leq i<n$. When $(u,v)$ runs through $\mathbb{F}_{p^n}\times\mathbb{F}_{p^n}^*$, $W_d(u,v)$ takes value in the set

\begin{equation*}
\begin{aligned}
\bigg\{ &0,\ 
3\eta_{\psi^j} + 1,\ 
2\eta_{\psi^j} + \eta_{\psi^{j+1}} + 1,\ 
2\eta_{\psi^j} + \eta_{\psi^{j+2}} + 1,\ 
2\eta_{\psi^j} + \frac{p^n + 2}{3} \ \bigg| \ j = 0, 1, 2 \bigg\},
\end{aligned}
\end{equation*}
if $\frac{1}{3}p^{-i}({p^{n}-1})\equiv 1\,\,({\rm mod}\,\,3)$ and
\begin{equation*}
\begin{aligned}
\bigg\{ &0,\ 
3\eta_{\psi^j} + 1,\ 
2\eta_{\psi^j} + \eta_{\psi^{j+1}} + 1,\ 
2\eta_{\psi^j} + \eta_{\psi^{j+2}} + 1,\ 
\eta_{\psi^j} +\eta_{\psi^{j+1}}+ \frac{p^n + 2}{3} \ \bigg| \ j = 0, 1, 2 \bigg\},
\end{aligned}
\end{equation*}
if $\frac{1}{3}p^{-i}({p^{n}-1})\equiv 0\,\,({\rm mod}\,\,3)$.
\end{lemma}
\begin{proof}
Note that every $x \in \mathbb{F}_{p^n}$ can be written as $x={\psi}^jy^3$ for some $y \in \mathbb{F}_{p^n}$ and $0\leq j\leq2$. From $y^{3d}=y^{3p^i}$ for $y \in \mathbb{F}_{p^n}$,
\begin{equation}\label{1mod3 s}
  \begin{aligned}
  3W_{d}(u,v)&=\sum\limits_{j=0}^{2}\sum\limits_{y \in \mathbb{F}_{p^n}}\zeta_p^{\Tr_1^{n}\left(u{\psi}^jy^3-v{({\psi}^jy^3)}^{d}\right)}\\
  &=\sum\limits_{j=0}^{2}\sum\limits_{y \in \mathbb{F}_{p^n}}\zeta_p^{\Tr_1^{n}\big({u{\psi}^jy^3-(v^{p^{-i}}{\psi}^{djp^{-i}}y^3)}^{p^i}\big)}\\
  &=\sum\limits_{j=0}^{2}\sum\limits_{y \in \mathbb{F}_{p^n}}\zeta_p^{\Tr_1^{n}\left(u{\psi}^jy^3-v^{p^{-i}}{\psi}^{djp^{-i}}y^3\right)}\\
  &=\sum\limits_{y \in \mathbb{F}_{p^n}}\zeta_p^{\Tr_1^{n}\left(y^3\left(u-v^{p^{-i}}\right)\right)}
  +\sum\limits_{y \in \mathbb{F}_{p^n}}\zeta_p^{\Tr_1^{n}\left(y^3\left(u-v^{p^{-i}}{\beta}\right)\psi\right)}
  +\sum\limits_{y \in \mathbb{F}_{p^n}}\zeta_p^{\Tr_1^{n}\left(y^3\left(u-v^{p^{-i}}{\beta}^2\right){\psi}^2\right)},\\
\end{aligned}
\end{equation}
where $\beta={\psi}^{p^{-i}(p^n-1)/3}$. Denote
\begin{equation*}
  \mathcal{G}:=\big\{(u,v)\in\mathbb{F}_{p^n}\times\mathbb{F}_{p^n} \mid (u-v^{p^{-i}}{\beta}^j)\neq0, \mbox{ for } j=0,1,2\big\}.
\end{equation*}

(i) Let $(u,v)\in\mathcal{G}$. According to (\ref{1mod3 s}), $3W_{d}(u,v)$ can take on at most ten values:
\[\eta_{1}+\eta_{\psi}+\eta_{\psi^2}+1=0, 3\eta_1+1, 3\eta_{\psi}+1, 3\eta_{\psi^2}+1, 2\eta_{1}+\eta_{\psi}+1,\] 
\[2\eta_{\psi}+\eta_{\psi^2}+1, 2\eta_{\psi^2}+\eta_{1}+1, 2\eta_{1}+\eta_{\psi^2}+1, 2\eta_{\psi}+\eta_{1}+1, 2\eta_{\psi^2}+\eta_{\psi}+1.\]

(ii) Let $(u,v)\not\in\mathcal{G}$ and $v\neq0$. Put $\operatorname{ind}_{\psi}(1-\beta_1)=r$, $\operatorname{ind}_{\psi}(0)=-\infty$ and $g=p^{-i}(p^n-1)/3=\operatorname{ind}_{\psi}(\beta)$. Since $1+\beta=-\beta^2$, $-1\in\mathcal{D}_0$, then 
\[\operatorname{ind}_{\psi}(1+\beta)=\operatorname{ind}_{\psi}(-\beta^2)\equiv 2g\,\,({\rm mod}\,\,3).\] 
For fixed $v\in\mathbb{F}_{p^n}^*$, let $\operatorname{ind}_{\psi}\big(v^{p^{-i}}\big)=s$, the indices of $u-v^{p^{-i}}$, $(u-v^{p^{-i}}{\beta})\psi$ and $(u-v^{p^{-i}}{\beta}^2)\psi^2$ then become:
\begin{table}[H]
\centering
 \caption{The indices of $u-v^{p^{-i}}$, $(u-v^{p^{-i}}{\beta})\psi$ and $(u-v^{p^{-i}}{\beta}^2)\psi^2$}
\label{ind table}
\begin{tabular}{|l|l|l|l|l|}
\hline
 & $\operatorname{ind}_{\psi}\big(u-v^{p^{-i}}\big)$ & $\operatorname{ind}_{\psi}\big((u-v^{p^{-i}}{\beta})\psi\big)$& 
 $\operatorname{ind}_{\psi}\big((u-v^{p^{-i}}{\beta^2})\psi^2\big)$\\
[0.5ex]
\hline
$u=v^{p^{-i}}$ & $-\infty$ & $s+r+1$  & $s+2g+r+2$ \\
\hline
$u=v^{p^{-i}}\beta$ & $s+r$ & $-\infty$  & $s+g+r+2$ \\
\hline
$u=v^{p^{-i}}\beta^2$ & $s+r+2g$ & $s+g+r+1$  & $-\infty$ \\
\hline
\end{tabular}
\end{table}
For fixed $v\in\mathbb{F}_{p^n}^*$, according to $\eta_0=\frac{p^n-1}{3}$, (\ref{1mod3 s}) and Table \ref{ind table}, when $g=\frac{1}{3}p^{-i}({p^{n}-1})\equiv 1\,\,({\rm mod}\,\,3)$,
\begin{equation}\label{oncetime}
W_d(u,v) = 2\eta_{\psi^j} + \frac{p^n + 2}{3} \quad \text{occurs once for any } j=0,1,2;
\end{equation}
when $g=\frac{1}{3}p^{-i}({p^{n}-1})\equiv 0\,\,({\rm mod}\,\,3)$,
\[
W_d(u,v) = \eta_{\psi^j} + \eta_{\psi^{j+1}} + \frac{p^n + 2}{3}\quad \text{occurs once for any } j=0,1,2.
\]
\end{proof}

\begin{remark}\rm
The values \(\eta_{\psi^j}\) (\(j = 0, 1, 2\)) in Lemma \ref{wf21value} are Gaussian periods, which are difficult to compute and rarely evaluated explicitly. When \(p \equiv 1 \pmod{3}\), the values of \(\eta_1\) were studied in \cite[p. 422]{BEW}, and their accurate values depend on some angle (see Lemma \ref{gauss3}). We shall later present an alternative form of \(\eta_{\psi^j}\) \((j = 0, 1, 2)\). Although this new version still does not uniquely determine the exact values, it is sufficient to determine \(C_{\tau}(s,s')\).
\end{remark}

\begin{lemma}\rm\label{gauss3}
Let \(p\) be an odd prime with \(p\equiv 1\,\,({\rm mod}\,\,3)\), and let \(n\) be a positive integer. Set \(q = p^n\). Let $\psi$ be a primitive element of $\mathbb{F}_{q}$.  Then
\[
\left\{\eta_{\psi^j} \ \bigg| \ j = 0, 1, 2 \right\}
=\left\{ \frac{2(-1)^{n+1}p^{n/2}}{3}   \cos\left( \frac{n\theta}{3} - \frac{2\pi l}{3} \right)-\frac{1}{3} \ \bigg| \ l = 0, 1, 2 \right\},
\]
\[
\left\{2\eta_{\psi^j}+\eta_{\psi^{j+1}} \ \bigg| \ j = 0, 1, 2 \right\}
=\left\{ \frac{2\sqrt{3}(-1)^{n+1}p^{n/2}}{3} \cos\left( \frac{n\theta}{3} -\frac{\pi(4l+1)}{6} \right) - 1 \ \bigg| \ l = 0, 1, 2 \right\},
\]
\[
\left\{2\eta_{\psi^j}+\eta_{\psi^{j+2}} \ \bigg| \ j = 0, 1, 2 \right\}
=\left\{ \frac{2\sqrt{3}(-1)^{n+1}p^{n/2}}{3} \cos\left( \frac{n\theta}{3} -\frac{\pi(4l-1)}{6} \right) - 1 \ \bigg| \ l = 0, 1, 2 \right\},
\]
\[
\left\{\eta_{\psi^j}+\eta_{\psi^{j+1}} \ \bigg| \ j = 0, 1, 2 \right\}
=\left\{\frac{2(-1)^{n+1}p^{n/2}}{3} \cos\left( \frac{n\theta}{3} - \frac{\pi(2l+1)}{3} \right)-\frac{2}{3} \ \bigg| \ l = 0, 1, 2 \right\},
\]
where $\theta =  \operatorname{sgn}(v) \cdot \arccos\frac{u}{2\sqrt{p}}$ and
 \(u, v\) are the integers uniquely determined by
\[
  u^2 + 3v^2=4p,  u \equiv 1\,\,({\rm mod}\,\,3), v\equiv 0\,\,({\rm mod}\,\,3) \mbox{ and } 3v\equiv u(2\psi^{(p^n-1)/3}+1)\,\,({\rm mod}\,\,p).
\]
\end{lemma}
\begin{proof}
Let  \(e(x)=\zeta_p^{\Tr_1^n(x)}\). Let $\psi$ be a primitive element of $\mathbb{F}_{q}$. Let \(\chi\) be the character on \(\mathbb{F}_q\) of order 3 given by $\chi(\psi)=\omega=e^{2\pi i/3}$. As \(\alpha\) runs through \(\mathbb{F}_q^*\), \(\alpha^3\) runs through the set of nonzero cubes in \(\mathbb{F}_q\), each occurring three times.
Hence
\begin{align*}
\sum\limits_{x\in \mathbb{F}_{q}}\zeta_p^{\Tr_1^n(\psi^j x^3)}
&=\sum_{\alpha \in \mathbb{F}_q} e(\psi^j\alpha^3) = 1 + \sum_{\alpha \in \mathbb{F}_q^*} e(\psi^j\alpha) \left( 1 + \chi(\alpha) + \chi^2(\alpha) \right)\\
& = \sum_{k=1}^{2} G_n(\psi^j,\chi^k) = \sum_{k=1}^{2} \chi^k(\psi^{-j}) G_n(\chi^k),
\end{align*}
where \(G_n(\psi^j, \chi^k) = \sum_{\alpha \in \mathbb{F}_q^*} e(\psi^j\alpha) \chi^k(\alpha)\) is the Gauss sum, and the last equality follows from \(G_n(\psi^j, \chi^k) = \chi^k(\psi^{-j}) G_n(\chi^k)\). Here \(\chi(\psi^{-1}) = \overline{\chi}(\psi)\), and since \(\chi\) is a character of order 3, we have \(\chi^2(\psi^{-1}) = \chi(\psi)\). Therefore,
\[
\sum_{\alpha \in \mathbb{F}_q} e(\psi^j\alpha^3) = \overline{\chi}(\psi^j) G_n(\chi) + \chi(\psi^j) G_n(\chi^2).
\]
Let \(\chi^*\) denote the restriction of \(\chi\) to \(\mathbb{F}_p\). By \cite[Proposition 11.4.1]{BEW}, the order of the restricted character \(\chi^*\) is
\begin{equation*}
    k^* = \frac{3}{\gcd(3, (q-1)/(p-1))} = \frac{3}{\gcd(3, 1+p+\cdots+p^{n-1})}=
	\left\{ \begin{array}{llll}
		1 &\mbox{if}& n \equiv 0\,\,({\rm mod}\,\,3),\\
		3 &\mbox{if}& n \not\equiv 0\,\,({\rm mod}\,\,3).
	\end{array}\right.
\end{equation*}

{\textbf{Case 1:}}  \(n \equiv 0 \pmod{3}\). It follows that  \(k^* = 1\).  For \(k = 1, 2\), applying \cite[Theorem 12.1.1]{BEW}, we have
\[
G_n(\chi^k) = -p E_n(\chi^k),
\]
where \(E_n(\chi^k) = \sum\limits_{a \in \mathbb{F}_{q}^*,  \Tr_1^n(a)=1} \chi^k(a)\) denotes the Eisenstein sum on \(\mathbb{F}_q\) corresponding to \(\chi^k\). Consequently,
\[
 \sum_{\alpha \in \mathbb{F}_q} e(\psi^j\alpha^3) =-p \overline{\chi}(\psi^j) E_n(\chi) -p\chi(\psi^j)  E_n(\chi^2).
\]
Since \(\chi\) has order 3, \(\chi^2 = \overline{\chi}\) and \(E_n(\chi^2) = \overline{E_n(\chi)}\). Therefore,
\[
 \sum_{\alpha \in \mathbb{F}_q} e(\psi^j\alpha^3) = -2p\cdot \mathrm{Re} \left( \overline{\chi}(\psi^j) E_n(\chi) \right).
\]
By  \cite[Theorem 12.3.1]{BEW}, for $\chi(\psi)=\omega=e^{2\pi i/3}$, \(p \equiv 1 \pmod{3}\) and \(n \equiv 0 \pmod{3}\), the Eisenstein sum \(E_n(\chi)\) is given by
\[
E_n(\chi) = (-1)^{n} p^{n/3-1} \left( \frac{1}{2} (u + i v \sqrt{3}) \right)^{n/3},
\]
where the integers \(u, v\) satisfy 
\[
  u^2 + 3v^2=4p,  u \equiv 1\,\,({\rm mod}\,\,3), v\equiv 0\,\,({\rm mod}\,\,3) \mbox{ and } 3v\equiv u(2\psi^{(p^n-1)/3}+1)\,\,({\rm mod}\,\,p).
\] 
Then, we obtain
\begin{align*}
\sum_{\alpha \in \mathbb{F}_q} e(\psi^j\alpha^3)   
&=  -2p \cdot \mathrm{Re}  \left( \overline{\chi}(\psi^j) E_n(\chi) \right) \\
&= (-1)^{n+1} 2p^{n/3}\cdot \mathrm{Re}  \left( \omega^{-j}    \left( \frac{1}{2} (u + i v \sqrt{3}) \right)^{n/3} \right) .
\end{align*}

{\textbf{Case 2:}}  \(n \not\equiv 0 \pmod{3}\).
For \(k = 1, 2\), applying \cite[Theorem 12.1.1]{BEW}, the restricted character \((\chi^k)^* = (\chi^*)^k\) is nontrivial, and we have
\[
G_n(\chi^k) = G((\chi^*)^k) E_n(\chi^k).
\]
Consequently,
\[
 \sum_{\alpha \in \mathbb{F}_q} e(\psi^j\alpha^3) = \overline{\chi}(\psi^j)G(\chi^*) E_n(\chi) +\chi(\psi^j) G(\chi^{*2}) E_n(\chi^2)= 2 \cdot \mathrm{Re}  \left( \overline{\chi}(\psi^j) G(\chi^*) E_n(\chi) \right). 
\]
We now evaluate the two factors \(G(\chi^*)\) and \(E_n(\chi)\). First, by  \cite[Theorem 4.1.6]{BEW}, the cubic Gauss sum \(G(\chi^*)\) on \(\mathbb{F}_p\) is given by

\begin{equation}\label{gausseisen}
G(\chi^*) = \frac{1}{2} p^{1/3} \left(r_3 + i s_3 \sqrt{3}\right) H(\chi^*), 
\end{equation}
where the integers \(u, v\) satisfy 
\[
  u^2 + 3v^2=4p,  u \equiv 1\,\,({\rm mod}\,\,3), v\equiv 0\,\,({\rm mod}\,\,3) \mbox{ and } 3v\equiv u(2\psi^{(p^n-1)/3}+1)\,\,({\rm mod}\,\,p),
\] 
and \(H(\chi^*)\) is the product constructed by Cassels in \cite{Cassels} (see also  \cite[(4.1.4)]{BEW}), which satisfies
\[
H^3(\chi^*) = J^{-2}(\chi^*, \chi^*),
\]
where $J^{-2}(\chi^*, \chi^*)= \left( \frac{1}{2}(u + i v\sqrt{3}) \right)^{-2}$ is a Jacobi sum.

\noindent\textbf{Subcase 1: } \(n \equiv 1 \pmod{3}\).
 By  \cite[Theorem 12.3.1]{BEW}, for $\chi(\psi)=\omega=e^{2\pi i/3}$, \(p \equiv 1 \pmod{3}\) and \(n \equiv 1 \pmod{3}\), the Eisenstein sum \(E_n(\chi)\) is given by
\begin{equation}\label{Enchi1}
E_n(\chi) = (-1)^{n-1} p^{(n-1)/3} \left( \frac{1}{2} (u + i v \sqrt{3}) \right)^{(n-1)/3}.
\end{equation}
Substituting (\ref{gausseisen}) and (\ref{Enchi1}) into the product \(G(\chi^*) E_n(\chi)\), we obtain
\begin{align*}
\sum_{\alpha \in \mathbb{F}_q} e(\psi^j\alpha^3)   
&=  2 \cdot \mathrm{Re}  \left( \overline{\chi}(\psi^j) G(\chi^*) E_n(\chi) \right) \\
&= 2 \cdot \mathrm{Re}  \left(\overline{\chi}(\psi^j) \left( p^{1/3} \frac{1}{2}(u + i v\sqrt{3})H(\chi^*) \right) \left( (-1)^{n-1} p^{(n-1)/3} \left( \frac{1}{2}(u + i v\sqrt{3}) \right)^{(n-1)/3} \right) \right) \\
&=2(-1)^{n+1}p^{n/3} \cdot \mathrm{Re} \left( \overline{\chi}(\psi^j) \omega^{j_1}J^{-2/3}(\chi^*, \chi^*) \left( \frac{1}{2}(u + i v\sqrt{3}) \right)^{(n+2)/3} \right)\\
&=2(-1)^{n+1}p^{n/3} \cdot \mathrm{Re}  \left( \omega^{-j}  \omega^{j_1}  \left( \frac{1}{2}(u + i v\sqrt{3}) \right)^{-2/3} \left( \frac{1}{2}(u + i v\sqrt{3}) \right)^{(n+2)/3} \right)\\
&=2(-1)^{n+1}p^{n/3} \cdot \mathrm{Re} \left( \omega^{-j} \omega^{j_1}   \left( \frac{1}{2}(u + i v\sqrt{3}) \right)^{n/3} \right)
\end{align*}
holds for some $j_1 \in \{0,1,2\}$.

\noindent\textbf{Subcase 2: } \(n \equiv 2 \pmod{3}\). By \cite[Theorem 12.3.1]{BEW},  for $\chi(\psi)=\omega=e^{2\pi i/3}$, \(p \equiv 1 \pmod{3}\) and \(n \equiv 2 \pmod{3}\), the Eisenstein sum \(E_n(\chi)\) is given by
\begin{equation}\label{Enchi2}
E_n(\chi) = (-1)^{n-1} p^{(n-2)/3} \left( \frac{1}{2}(u + iv\sqrt{3}) \right)^{(n+1)/3}.
\end{equation}
Substituting (\ref{gausseisen}) and (\ref{Enchi2}) into the product \(G(\chi^*) E_n(\chi)\), we obtain
\begin{align*}
\sum_{\alpha \in \mathbb{F}_q} e(\psi^j\alpha^3)   
&=  2 \cdot \mathrm{Re}  \left( \overline{\chi}(\psi^j) G(\chi^*) E_r(\chi) \right) \\
&= 2 \cdot \mathrm{Re} \left(\overline{\chi}(\psi^j) \left( p^{1/3} \frac{1}{2}(u + i v\sqrt{3})H(\chi^*) \right) \left( (-1)^{n-1} p^{(n-2)/3} \left( \frac{1}{2}(u + i v\sqrt{3}) \right)^{(n+1)/3} \right) \right) \\
&= 2(-1)^{n+1}p^{(n-1)/3} \cdot \mathrm{Re}  \left( \overline{\chi}(\psi^j) \omega^{j_1} J^{-2/3}(\chi^*, \chi^*) \left( \frac{1}{2}(u + i v\sqrt{3}) \right)^{(n+4)/3} \right)\\
&= 2(-1)^{n+1}p^{(n-1)/3} \cdot \mathrm{Re}  \left( \omega^{-j} \omega^{j_1} \left( \frac{1}{2}(u + i v\sqrt{3}) \right)^{-2/3}  \left( \frac{1}{2}(u + i v\sqrt{3}) \right)^{(n+4)/3} \right)\\
&=2(-1)^{n+1}p^{n/3} \cdot \mathrm{Re}  \left( \omega^{-j} \omega^{j_1}   \left( \frac{1}{2}(u + i v\sqrt{3}) \right)^{n/3} \right).
\end{align*}
holds for some $j_1 \in \{0,1,2\}$.

Based on the discussion of two cases above, when $p \equiv 1 \pmod{3}$, we always have
\begin{align*}
\sum_{\alpha \in \mathbb{F}_q} e(\psi^j\alpha^3)   = (-1)^{n+1} 2p^{n/3} \operatorname{Re} \left(\omega^{-j}  \omega^{j_1} \left( \frac{1}{2} (u + i v \sqrt{3}) \right)^{n/3} \right),
\end{align*}
which holds for some $j_1 \in \{0,1,2\}$.
Set
\[
Z = \frac{1}{2}\left(u + i v\sqrt{3}\right),
\]
which has modulus $|Z| = \sqrt{p}$ and argument $\theta = \operatorname{sgn}(v) \cdot \arccos\left(\frac{u}{2\sqrt{p}}\right)$.
However, when $n \not\equiv 0 \pmod{3}$, the expression $Z^{n/3}$ is multi-valued and has three distinct possibilities, namely,
\[
Z^{n/3} \in \left\{ p^{n/6} e^{i n\theta/3},\; e^{2\pi i/3} p^{n/6} e^{i n\theta/3},\; e^{4\pi i/3} p^{n/6} e^{i n\theta/3} \right\}.
\]
Consequently, for some $l \in \{0,1,2\}$, we have
\begin{align*}
\sum_{\alpha \in \mathbb{F}_q} e(\psi^j\alpha^3)   
&= 2(-1)^{n+1} p^{n/3} \cdot \operatorname{Re} \left( \omega^{-j} \omega^{j_1}  \left( \frac{1}{2}(u + i v\sqrt{3}) \right)^{n/3} \right)\\
&= 2(-1)^{n+1} p^{n/3}\cdot  \operatorname{Re} \left( \omega^{-l} p^{n/6} e^{i n\theta/3} \right)\\
&= 2(-1)^{n+1} p^{n/2} \cos\left( \frac{n\theta}{3} - \frac{2\pi l}{3} \right),
\end{align*}
and \begin{align*}
\sum_{\alpha \in \mathbb{F}_q} e(\psi^{j+1}\alpha^3)   
= 2(-1)^{n+1} p^{n/2} \cos\left( \frac{n\theta}{3} - \frac{2\pi (l+1)}{3} \right).
\end{align*}
\end{proof}

\begin{proof}[Proof of Theorem~\ref{13distribution}]
We just calculate the conclusion for $\frac{1}{3}p^{-i}({p^{n}-1})\equiv 1\,\,({\rm mod}\,\,3)$. When $(u,v)$ runs through ${\mathbb{F}_{p^n}} \times {\mathbb{F}_{p^n}}$, the possible values of $W_d(u,v)$ are confirmed in Lemma \ref{wf21value}. It follows from (\ref{oncetime}) that $2\eta_{\psi^s}+\frac{p^n-1}{3}$ occurs $p^n-1$ times for $s=0,1,2$.
Define
\begin{equation}\label{Ni}
  G_{j,k,l} = \left| \left\{(u,v) \in{\mathbb{F}_{p^n}} \times {\mathbb{F}_{p^n}} : W(u,v)=\eta_{\psi^j}+\eta_{\psi^k}+\eta_{\psi^l}+1 \right\}\right|, \mbox{ for } j,k,l\in \{0,1,2\}.
\end{equation}
By Lemma \ref{wf21value}, we have 
\begin{equation*}
  \begin{aligned}
  3W_{d}(u,v)=\sum\limits_{y \in \mathbb{F}_{p^n}}\zeta_p^{\Tr_1^{n}\left(y^3(u-v^{p^{-i}})\right)}
  +\sum\limits_{y \in \mathbb{F}_{p^n}}\zeta_p^{\Tr_1^{n}\left(y^3(u-v^{p^{-i}}{\beta})\psi\right)}
  +\sum\limits_{y \in \mathbb{F}_{p^n}}\zeta_p^{\Tr_1^{n}\left(y^3(u-v^{p^{-i}}{\beta}^2){\psi}^2\right)},\\
\end{aligned}
\end{equation*}
where $\beta={\psi}^{p^{-i}(p^n-1)/3}$.
Denote
\begin{equation*}
  \mathcal{G}:=\left\{(u,v)\in\mathbb{F}_{p^n}\times\mathbb{F}_{p^n} \mid (u-v^{p^{-i}}{\beta}^j)\neq0, \mbox{ for } j=0,1,2\right\}.
\end{equation*}
For each $(u,v)\in\mathcal{G}$, we define the function $f_{u,v}:\{0,1,2\}\rightarrow\{0,1,2\}$ based on the rule $f_{u,v}(j)=r$ with $r \equiv \operatorname{ind}_{\psi}\big(\psi^j(u-v^{p^{-i}}{\beta}^j)\big)\,\,({\rm mod}\,\,3)$. We then partition the set $\mathcal{G}$ into the following disjoint subsets:
\begin{equation*}
  S_{j,k,l}:=\big\{(u,v)\in\mathcal{G} \mid \big(f_{u,v}(0),f_{u,v}(1),f_{u,v}(2)\big)=(j,k,l) \big\}, \mbox{ for } j,k,l\in\{0,1,2\}.
\end{equation*}
For any \(u, v \in \mathbb{F}_{p^n}\) and \(0 \leq j \leq 2\), define \(u_j := \psi^j(u - v^{p^{-i}}\beta^j)\). Then
\begin{equation}\label{u0u1u2}
u_0 + \psi^{-1}\beta u_1 + \psi^{-2}\beta^2 u_2 = 0.
\end{equation}
Moreover, if we arbitrarily choose values for \(u_1\) and \(u_2\), there exists a unique \((u, v) \in \mathbb{F}_{p^n}^2\) such that
\[
u_1 = \psi(u - v^{p^{-i}}\beta), \quad
u_2 = \psi^2(u - v^{p^{-i}}\beta^2),
\]
and by (\ref{u0u1u2}) we further obtain
\[
u_0 = -\frac{u_1}{\psi\beta^2}\left(\frac{u_2}{\psi\beta^2 u_1} + 1\right).
\]
Hence, to compute (for example) \(|S_{0,0,0}|\), we only need to count how many pairs \((u_1, u_2) \in \mathbb{F}_{p^n}^2\) satisfy \(\operatorname{ind}_{\psi}(u_j) \equiv 0 \pmod{3}\) for all \(0 \leq j \leq 2\). We may assume that \(u_2\) takes any value in \(\mathcal{D}_0\), which gives \(\frac{p^n-1}{3}\) possible choices for \(u_2\). Now fix a \(u_2 \in \mathcal{D}_0\). Note that \(\psi\beta^2 \in \mathcal{D}_0\), \(-1 \in \mathcal{D}_0\), and \(u_0 = -\frac{u_1}{\psi\beta^2}\bigl(\frac{u_2}{\psi\beta^2 u_1} + 1\bigr)\). Therefore, for \(u_0\) and \(u_1\) to also belong to \(\mathcal{D}_0\), it is necessary that
\[
\frac{u_2}{\psi\beta^2 u_1} + 1 \in \mathcal{D}_0,
\]
and by Lemma \ref{crslemma}, the number of such choices for \(u_1\) is
\[
|C_{0,0}| = \frac{1}{9}\left(p^n + \frac{{(-1)}^{n-1}}{2^{n-1}}E(u,v,n) - 8\right).
\]
Consequently,
\begin{equation*}
G_{0,0,0} = |S_{0,0,0}| = \frac{p^n-1}{3}|C_{0,0}| = \frac{1}{27}(p^n-1)\bigg(p^n + \frac{{(-1)}^{n-1}}{2^{n-1}}E(u,v,n) - 8\bigg).
\end{equation*}
Similarly, we obtain
\begin{equation*}
  G_{1,1,1}=|S_{1,1,1}|=G_{2,2,2}=|S_{2,2,2}|
  =\frac{1}{27}(p^n-1)\bigg(p^n+\frac{{(-1)}^{n-1}}{2^{n-1}}E(u,v,n)-8\bigg),
\end{equation*} 
\begin{align*}
 G_{0,1,2}&=|S_{0,1,2}|+|S_{0,2,1}|+|S_{1,0,2}|+|S_{1,2,0}|+|S_{2,0,1}|+|S_{2,1,0}|\\
 &= \frac{p^n-1}{3}\left(|C_{2,1}|+|C_{1,2}|+|C_{1,2}|+|C_{2,1}|+|C_{2,1}|+|C_{1,2}|\right) \\
&= \frac{2}{9}(p^n-1)\left(p^n+\frac{{(-1)}^{n-1}}{2^{n-1}}E(u,v,n)+1\right),
\end{align*}
\begin{align*}
 G_{0,0,1}&=G_{1,1,2}=G_{2,2,0}= \frac{p^n-1}{3}\big( |C_{0,1}|+|C_{2,2}|+|C_{1,0}| \big) \\
&= \frac{1}{9}(p^n-1)\left(p^n+\frac{{(-1)}^{n}}{2^{n}}\big(E(u,v,n)-O(u,v,n)\big)-2\right),
\end{align*}
\begin{align*}
 G_{0,0,2}&=G_{1,1,0}=G_{2,2,1}= \frac{p^n-1}{3}\big( |C_{0,2}| +|C_{2,0}|+ |C_{1,1}| \big)\\ &=\frac{1}{9}(p^n-1)\left(p^n+\frac{{(-1)}^{n}}{2^{n}}\big(E(u,v,n)+O(u,v,n)\big)-2\right).
\end{align*}

Since $\gcd(d,p^n-1)=1$, we have $W_{d}(0,v)=0$ for any $v\in \mathbb{F}_{p^n}^*$. Note that for each fixed $v\in \mathbb{F}_{p^n}^*$, the value distribution of $W_{d}(u/v^{d^{-1}},1)$ as $u$ runs through $\mathbb{F}_{p^n}^*$ is the same as that of $W_{d}(u,1)$ when $u$ runs through $\mathbb{F}_{p^n}^*$, where $d^{-1}$ denotes the inverse of $d$ modulo $p^n-1$. Moreover, for each fixed $v \in \mathbb{F}_{p^n}^*$, $W_{d}(u,v)=W_{d}(u/v^{d^{-1}},1)$. It is also obvious that $W_{d}(u,0)=0$ for any $u\in \mathbb{F}_{p^n}^*$. Therefore, the value distribution of $W_{d}(u,1)=C_{\tau}(s,s')+1$ as $u$ runs through $\mathbb{F}_{p^n}^*$ can be obtained.
\end{proof}

In what follows, we provide two examples to verify our results.
\begin{example} \rm
Let $p=7$, $n=5$, $d=\frac{7^{5}-1}{3}+1$. One can check that 3 is a primitive element of $\mathbb{F}_7$.  Let $\psi$ be a primitive element of $\mathbb{F}_{7^5}$ such that $\psi^{\frac{7^5-1}{7-1}}=3$.  Then $u=1$, $v=-3$, $E(1,-3,5)=3376$, $O(1,-3,5)=-4176$, $\theta =  - \arccos\frac{\sqrt{7}}{14}$. The cross correlation distribution $C_{\tau}(s,s')$ when $\tau$ runs from $0$ to $p^n-2$ can be described as follows.
\begin{table}[H]
\centering
\begin{tabular}{lll}
\hline
Value & Frequency \\
[0.5ex]
\hline
-1 & 3780 \\
$98\sqrt{7} \cos\left( \frac{5\theta}{3} - \frac{2\pi h}{3} \right) - 1$
& 630  for $h=0,1,2$   \\ 
$\frac{98\sqrt{21}}{3} \cos\left( \frac{5\theta}{3} -\frac{(4h+1)\pi}{6} \right) - 1$
 & 1841 for $h=0,1,2$ \\ 
$ \frac{98\sqrt{21}}{3} \cos\left( \frac{5\theta}{3} - \frac{(4h-1)\pi}{6} \right) - 1$ &  1870  for $h=0,1,2$ \\
$\frac{16804}{3} + \frac{196\sqrt{7}}{3} \cos\left( \frac{5\theta}{3} - \frac{2\pi h}{3} \right)$ &  1  for $h=0,1,2$ \\
\hline
\end{tabular}
\end{table}
\end{example}

\begin{example}\rm
 Let $p=13$, $n=3$, $d=\frac{13^{3}-1}{3}+1$. One can check that 2 is a primitive element of $\mathbb{F}_{13}$.  Let $\psi$ be a primitive element of $\mathbb{F}_{13^3}$ such that $\psi^{\frac{13^3-1}{13-1}}=2$.  Then $u=-5$, $v=-3$, $E(-5,-3,3)=280$, $O(-5,-3,3)=-432$, $\theta =  -\arccos\frac{-5\sqrt{13}}{26}$. The cross correlation distribution $C_{\tau}(s,s')$ when $\tau$ runs from $0$ to $p^n-2$ can be described as follows.
\begin{table}[H]
\centering
\begin{tabular}{l l l}
\hline
Value & Frequency (each) \\
[0.5ex]
\hline
-1 & 501 \\
-27, 90, -66 & 84  \\ 
38, -53, 12 &  234 \\
51, -40, -14 & 246   \\ 
701, 753, 740 &  1 \\
\hline
\end{tabular}
\end{table}
\end{example}

As an application, we consider a class of cyclic codes. Let $h_1(x)$ and $h_d(x)$ be the minimal polynomials of $\psi^{-1}$ and $\psi^{-d}$ over $\mathbb{F}_{p}$, respectively. Let $\mathcal{C}_{1,d}$ be the cyclic code with parity-check polynomial $h_1(x)h_d(x)$. By Delsarte's Theorem \cite{Delsarte_theorem}, the cyclic code $\mathcal{C}_{1,d}$ can be expressed as
\begin{equation*}
    \mathcal{C}_{1,d}=\left\{c_{u,v}=\left(\Tr ^{n}_1(u\psi^{i}+v\psi^{id})\right)_{i=0}^{p^n-2}\mid u,v\in \mathbb{F}_{p^{n}} \right\}.
\end{equation*}

\begin{theorem}\label{wf2 code}
 The weight distribution of $\mathcal{C}_{1,d}$ is given in Table \ref{codetable1} when \(n \equiv 0 \pmod{3}\), and in Table \ref{codetable2} when \(n \not\equiv 0 \pmod{3}\).
\begin{table}[H]\footnotesize
\centering
 \caption{  \(n \equiv 0 \pmod{3}\) }
\label{codetable1}
\begin{tabular}{l l }
\hline
      Weight              &   Number of codewords   \\
[0.5ex]
\hline
0 &1\\
$p^{n-1}(p-1)$ & $\frac{p^n-1}{9}\big(2p^n+2\frac{{(-1)}^{n-1}}{2^{n-1}}E(u,v,n)-7\big)$ \\
$p^{n-1}(p-1)-\frac{(p-1)(-1)^{n+1}2p^{n/2}}{p} \cos\left( \frac{n\theta}{3} - \frac{2\pi l}{3} \right)$
& $\frac{p^n-1}{27}\big(p^n+\frac{{(-1)}^{n-1}}{2^{n-1}}E(u,v,n)+1\big)$ for any $l=0,1,2$  \\ 
$p^{n-1}(p-1)-\frac{(p-1)2\sqrt{3}(-1)^{n+1}p^{n/2}}{3p} \cos\left( \frac{n\theta}{3} -\frac{\pi(4l+1)}{6} \right)$
 & $\frac{p^n-1}{9}\big(p^n+\frac{{(-1)}^{n}}{2^{n}}\big(E(u,v,n)-O(u,v,n)\big)-2\big)$ for any $l=0,1,2$  \\ 
$p^{n-1}(p-1)-\frac{(p-1)2\sqrt{3}(-1)^{n+1}p^{n/2}}{3p} \cos\left( \frac{n\theta}{3} -\frac{\pi(4l-1)}{6} \right) $ 
 &   $\frac{p^n-1}{9}\big(p^n+\frac{{(-1)}^{n}}{2^{n}}\big(E(u,v,n)+O(u,v,n)\big)-2\big)$ for any $l=0,1,2$  \\
$(p-1)\left(p^{n-1}-\frac{2(-1)^{n+1}p^{n/2}}{3p} \cos\left( \frac{n\theta}{3} - \frac{\pi(2l+1)}{3} \right) +\frac{p^n}{3}\right)$ &  $p^n-1$ for any $l=0,1,2$\\
\hline
\end{tabular}
\end{table}

\begin{table}[H]\footnotesize
\centering
 \caption{  \(n \not\equiv 0 \pmod{3}\) }
\label{codetable2}
\begin{tabular}{l l }
\hline
      Weight              &   Number of codewords   \\
[0.5ex]
\hline
0 &1\\
$p^{n-1}(p-1)$ & $(p^{n}-2)(p^n-1)$ \\
$\frac{2}{3}(p-1)p^{n-1}$ &  $3(p^n-1)$\\
\hline
\end{tabular}
\end{table}
  \end{theorem}
\begin{proof}
{\textbf{Case 1:}}  \(n \equiv 0 \pmod{3}\). 
For the codeword $c_{u,v}$ of $\mathcal{C}_{1,d}$, we also calculate its Hamming weight by using exponential sum:
\begin{equation*}
    \begin{aligned}
        \omega_H(c_{u,v}) &=p^{n}-1-\# \left\{x\in\mathbb{F}_{p^{n}}^*\mid \Tr_1^{n}\left(ux+vx^{\frac{p^n-1}{3}+p^i}\right)=0 \right\}\\
        &=p^{n}-\frac{1}{p}\sum\limits_{y\in \mathbb{F}_{p}}\sum\limits_{x\in \mathbb{F}_{p^{n}}}\omega_p^{y{\Tr ^{n}_1\left(ux+vx^{\frac{p^n-1}{3}+p^i}\right)}}\\
        &=p^{n-1}(p-1)-\frac{1}{p}\sum\limits_{y\in \mathbb{F}_{p}^*}\sum\limits_{x\in \mathbb{F}_{p^{n}}}\omega_p^{\Tr ^{n}_1\left(uxy+v{(xy)}^{\frac{p^n-1}{3}+p^i}\right)}\\
        &=p^{n-1}(p-1)-\frac{p-1}{p}\sum\limits_{x\in \mathbb{F}_{p^{n}}}\omega_p^{\Tr ^{n}_1\left(ux+vx^{\frac{p^n-1}{3}+p^i}\right)}\\
        &=p^{n-1}(p-1)-\frac{p-1}{p}W_d(u,-v).
    \end{aligned}
\end{equation*}
Therefore, the weight distribution of cyclic code $\mathcal{C}_{1,d}$ is directly determined by the value distribution of the exponential sum $W_d(u,-v)$. If $u=0$, we have
\begin{equation*}
    \omega_H(c_{0,v})=
	\left\{ \begin{array}{llll}
		0&\mbox{if}&v=0,\\
		p^{n-1}(p-1)&\mbox{if}&v\neq0.
	\end{array}\right.
\end{equation*}
If $u\neq0$, the value distribution of $\omega_H(c_{u,v})$ follows from Theorem \ref{13distribution}.

{\textbf{Case 2:}}  \(n \not \equiv 0 \pmod{3}\). Let $p=3h+1$. For the codeword $c_{u,v}$ of $\mathcal{C}$, we also calculate its Hamming weight by using exponential sum:
\begin{equation*}
    \begin{aligned}
        \omega_H(c_{u,v}) &=p^{n}-1-\# \left\{x\in\mathbb{F}_{p^{n}}^*\mid \Tr_1^{n}\left(ux+vx^{\frac{p^n-1}{3}+p^i}\right)=0 \right\}\\
        &=p^{n}-\frac{1}{p}\sum\limits_{y\in \mathbb{F}_{p}}\sum\limits_{x\in \mathbb{F}_{p^{n}}}\omega_p^{y{\Tr ^{n}_1\left(ux+vx^{\frac{p^n-1}{3}+p^i}\right)}}\\
        &=p^{n-1}(p-1)-\frac{h}{p}\sum\limits_{j=0}^2\sum\limits_{x\in \mathbb{F}_{p^{n}}}\omega_p^{\Tr^{n}_1\left(ux+\beta^{j}v{x}^{\frac{p^n-1}{3}+p^i}\right)}\\
        &=p^{n-1}(p-1)-\frac{h}{p}\sum\limits_{j=0}^2W_d(u,-\beta^jv),
    \end{aligned}
\end{equation*}
where $\beta={\psi}^{p^{-i}(p^n-1)/3}$.
By Lemma \ref{wf21value}, we have 
\begin{equation*}
  \begin{aligned}
  &\sum\limits_{j=0}^2W_d(u,-\beta^jv) \\
  &=\sum\limits_{j=0}^2 \frac{1}{3}
  \left( \sum\limits_{y \in\mathbb{F}_{p^n}}
  \zeta_p^{\Tr_1^{n}\left(y^3\left(u+(\beta^jv)^{p^{-i}}\right)\right)}
  +\zeta_p^{\Tr_1^{n}\left(y^3\left(u+(\beta^jv)^{p^{-i}}{\beta_1}\right)\psi\right)}
  +\zeta_p^{\Tr_1^{n}\left(y^3\left(u+(\beta^jv)^{p^{-i}}{\beta_1}^2\right){\psi}^2\right)}
  \right)\\
  &=\sum\limits_{j=0}^2 \frac{1}{3}
  \left( \sum\limits_{y \in \mathbb{F}_{p^n}}
  \zeta_p^{\Tr_1^{n}\left(y^3\left(u+v^{p^{-i}}\right)\psi^j\right)}
  +\zeta_p^{\Tr_1^{n}\left(y^3\left(u+\beta v^{p^{-i}}\right)\psi^j\right)}
  +\zeta_p^{\Tr_1^{n}\left(y^3\left(u+\beta^2 v^{p^{-i}}\right){\psi}^j\right)}
  \right).
\end{aligned}
\end{equation*}
 Denote
\begin{equation*}
  \mathcal{G}:=\left\{(u,v)\in\mathbb{F}_{p^n}\times\mathbb{F}_{p^n} \mid (u+v^{p^{-i}}{\beta_1}^j)\neq0, \mbox{ for } j=0,1,2\right\}.
\end{equation*}
For each $(u,v)\in\mathcal{G}$, $\sum\limits_{j=0}^2W_d(u,-\beta^jv)=0$; for each $(u,v)\not\in\mathcal{G}$,  $(u,v)\in\mathcal{G}$, $\sum\limits_{j=0}^2W_d(u,-\beta^jv)=p^n$.
\end{proof}

\end{document}